\documentclass[conference]{IEEEtran}
\usepackage{subfigure}
\usepackage{amsmath}
\usepackage{graphicx,amssymb,lineno,bm,subfigure}
\usepackage{algorithm}
\usepackage{algorithmic}

\newtheorem{Def}{\textbf{Definition}}

\newtheorem{proof}{\textbf{Proof}}
\newtheorem{corol}{\textbf{Corollary}}
\newtheorem{lma}{\textbf{Lemma}}
\IEEEoverridecommandlockouts
\begin{document}

%
\title{\huge{Joint Task Offloading Scheduling and Transmit Power Allocation for Mobile-Edge Computing Systems}}
\author{\IEEEauthorblockN{Yuyi Mao{$^\dagger$}, Jun Zhang{$^\dagger$}, and Khaled B. Letaief{$^{\dagger\ast}$}, \emph{Fellow, IEEE}}
\IEEEauthorblockA{$^\dagger$Dept. of ECE, The Hong Kong University of Science and Technology, Clear Water Bay, Hong Kong\\
$^{\ast}$Hamad bin Khalifa University, Doha, Qatar\\
Email: \{ymaoac, eejzhang, eekhaled\}@ust.hk}
\thanks{This work is supported by the Hong Kong Research Grants Council under Grant No. 16200214.
}
}
\maketitle

\begin{abstract}
Mobile-edge computing (MEC) has emerged as a prominent technique to provide mobile services with high computation requirement, by migrating the computation-intensive tasks from the mobile devices to the nearby MEC servers. To reduce the execution latency and device energy consumption, in this paper, we jointly optimize task offloading scheduling and transmit power allocation for MEC systems with multiple independent tasks. A low-complexity sub-optimal algorithm is proposed to minimize the weighted sum of the execution delay and device energy consumption based on alternating minimization. Specifically, given the transmit power allocation, the optimal task offloading scheduling, i.e., to determine the order of offloading, is obtained with the help of flow shop scheduling theory. Besides, the optimal transmit power allocation with a given task offloading scheduling decision will be determined using convex optimization techniques. Simulation results show that task offloading scheduling is more critical when the available radio and computational resources in MEC systems are relatively balanced. In addition, it is shown that the proposed algorithm achieves near-optimal execution delay along with a substantial device energy saving.
\end{abstract}

%
\IEEEpeerreviewmaketitle
\section{Introduction}
The rapid development of mobile applications with advanced features has brought great pressure on mobile computing systems. However, the limited processing capability of mobile devices becomes an obstacle to fulfill such a requirement. Mobile-edge computing (MEC) has emerged as a promising technique to resolve this issue, which offers computation capability within the radio access network in contrast to conventional cloud computing systems that use remote public clouds \cite{ETSI1409,YMao17MECSurvey}. By offloading the computation-intensive tasks from the mobile devices to the nearby MEC servers, the quality of computation experience, including the latency and device energy consumption, could be greatly improved \cite{Kumar1302}.

Nevertheless, the efficiency of an MEC system largely depends on the adopted computation offloading policy, which should be carefully designed by taking the characteristics of the computation tasks and wireless channels into account \cite{WZhang1309}-\cite{Khalili15}. In \cite{WZhang1309}, energy consumption for mobile execution and computation offloading for single-user MEC systems was minimized via dynamic voltage frequency scaling and data transmission scheduling, respectively. A delay-optimal task scheduling algorithm for single-user MEC systems was proposed in \cite{JLiu1607}. For multi-user MEC systems, a decentralized computation offloading policy was proposed in \cite{XChen1504}, and joint sub-carrier and CPU time allocation was investigated in \cite{YYu1612}.

Inspired by the fact that many applications can be divided into a set of dependent sub-tasks, fine-grained computation offloading (a.k.a. code partitioning) has also been widely studied most recently \cite{MJia1404}-\cite{Khalili15}. Assuming serial implementation of communication and computation, a heuristic offloading algorithm was proposed in \cite{MJia1404}, while the joint optimization of radio resource and code partitioning was investigated in \cite{Lorenzo13}. Besides, a dynamic code partitioning algorithm was proposed in \cite{DHuang1206}, which is adaptive to the time-varying wireless data rate and the random application requests. In order to reduce the execution latency, parallel implementation of communication and computation was adopted in \cite{YHKao14}-\cite{Khalili15}, where the code partitioning policy design becomes much more challenging. In \cite{YHKao14}, a code partitioning algorithm was developed based on dynamic programming for applications with a tree topology. This study was extended to applications with general topologies in \cite{Mahmoodi16}, where the solution has exponential complexity. By leveraging the structures of the application topology, a code partitioning algorithm was proposed to reduce the computation complexity based on message passing in \cite{Khalili15}. However, these studies assume that each offloaded sub-task is allocated with certain amount of communication bandwidth and computation resource at the MEC server, i.e., transmitting the input data of multiple tasks (executing multiple tasks at the MEC server) concurrently is allowed. Although such assumption makes the design more tractable, it may be impractical for MEC systems with limited resources, e.g., for MEC servers with a single communication channel and a single-core CPU.

In this paper, we consider a single-user MEC system with multiple independent computation tasks requesting for mobile-edge execution, assuming parallel implementation of task offloading and execution. A radio resource limited system is considered, where each time the input data of only one task can be offloaded, while different tasks are executed by the MEC server sequentially. In this case, the transmission and execution processes of different tasks are coupled, and the task offloading scheduling, i.e., the order of the tasks to be offloaded, becomes a new design dimension. We formulate a joint task offloading scheduling and transmit power allocation problem with the objective of minimizing the weighted sum of the execution delay and device energy consumption, which is a highly non-trivial \emph{mixed integer nonlinear programming (MINLP)} problem. A low-complexity sub-optimal algorithm is proposed based on alternating minimization. In particular, the optimal task offloading scheduling with a given transmit power allocation is obtained with the aid of \emph{flow shop scheduling theory}, while the optimal transmit power allocation with a given task offloading scheduling decision is determined based on \emph{convex optimization}. Simulation results shall show that task offloading scheduling is more critical when the radio and computational resources are relatively balanced. Besides, near-optimal execution delay along with a substantial device energy saving can be achieved by the proposed algorithm.


\section{System Model}
\vspace{-10pt}
\begin{figure}[h]
\centering
\includegraphics[width=0.5\textwidth]{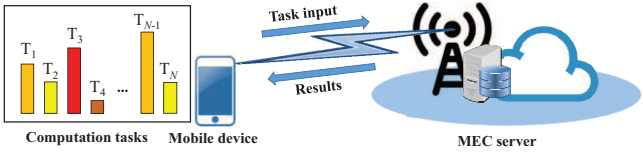}
\vspace{-20pt}
\caption{A mobile-edge computing system with a mobile device and an MEC server.}
\label{sysmodelMEC}
\end{figure}

We consider a single-user mobile-edge computing (MEC) system as shown in Fig. \ref{sysmodelMEC}, which consists of a mobile device and an MEC server. The MEC server is a small-scale data center deployed by the telecom operator and communicates with the mobile device through the wireless channel. A cloud-clone associated with the mobile device is running at the MEC server. Hence, the MEC server can execute the computation tasks on behalf of the mobile device \cite{Satyanarayanan0910}. Besides, the distance between the mobile device and the MEC server is denoted as $L$, and the system bandwidth is represented by $\omega$.

\subsection{Computation Task Model}

We assume that the mobile device has $N$ independent computation tasks that need to be executed, where the set of tasks is denoted as $\mathcal{T}\triangleq \{\text{T}_{1},\cdots,\text{T}_{N}\}$. Each computation task is characterized by a two-tuple of parameters, $\langle d_{i},c_{i}\rangle$, where $d_{i}$ (in bits) is the amount of task input data, and $c_{i}$ (in CPU cycles\slash bit) is the workload. The values of $d_{i}$ and $c_{i}$ depend on the nature of the computation tasks and can be obtained through off-line measurements \cite{Miettinen1006}. Denote $\mathbf{d}=\left[d_{1},\cdots,d_{N}\right]$ and $\mathbf{c}=\left[c_{1},\cdots,c_{N}\right]$. In this paper, we focus on the scenarios that the mobile device has very limited computational resource, and hence all the computation tasks should be offloaded to the MEC server for mobile-edge execution \cite{KWang16}. The scenarios that the mobile device has more powerful computation capability will be handled in our future work.

\subsection{Task Offloading and Mobile-Edge Execution Model}

In order to offload the computation tasks for mobile-edge execution, the task input data should be transmitted to the MEC server. We denote the task offloading scheduling decision for the $N$ tasks as $\bm{\sigma}\triangleq \left[\sigma_{1},\sigma_{2},\cdots,\sigma_{N}\right]$, which is a permutation of the task indices and $\text{T}_{\sigma_i}$ is the $i$th task being offloaded to the MEC server. Hence, $\bm{\sigma}$ satisfies $\sigma_{i}\in\{1,\cdots,N\}$ and $\sigma_{i}\neq \sigma_{j}, j\neq i,\forall i,j=1,\cdots,N$. The mobile device is equipped with a single antenna and it can transmit the input data of one computation task at each time. Thus, the transmission rate for task $\text{T}_{i}$ can be expressed as
\begin{equation}
R\left(p_{i}\right) = \omega \log_{2} \left(1+ \frac{g_{0}\left(L_{0}\slash L\right)^{\theta}p_{i}}{N_{0}\omega}\right),
\end{equation}
where $p_{i}$ is the transmit power for task $\text{T}_{i}$, $g_{0}$ is the path-loss constant, $\theta$ is the path-loss exponent, $L_{0}$ is the reference distance, and $N_{0}$ is the noise power spectral density at the receiver of the MEC server.

The MEC server has a single-core CPU and executes the offloaded tasks one after another. A task buffer is employed at the MEC server to store the tasks that have been offloaded but not yet executed, which is assumed to be sufficiently large for simplicity. Besides, it executes the computation task with a \emph{first-come-first-serve (FCFS)} fashion. Thus, the execution order of the tasks is the same as the task offloading order.  The CPU-cycle frequency at the MEC server is denoted as $f_{\rm{ser}}$ (in Hz). Besides, the computation results are assumed to be of small size and thus the feedback delay is ignored.

The execution delay and device energy consumption are two critical design considerations in MEC systems, and both of them depend on the adopted task offloading scheduling and transmit power allocation policy. In the next section, we will formulate an optimization problem to minimize the weighted sum of the execution delay and device energy consumption, by jointly designing the offloading scheduling and transmit power allocation for the computation tasks.

\section{Problem Formulation}

In this section, we will first analyze the execution delay and device energy consumption with given task offloading scheduling decision $\bm{\sigma}$ and transmit power allocation vector $\mathbf{p}=\left[p_{1},\cdots,p_{N}\right]$. The joint task offloading scheduling and transmit power allocation problem will then be formulated.

The mobile-edge execution for a computation task cannot be started until the following two conditions are satisfied: First, the task input data is ready at the MEC server. Second, the CPU at the MEC server is available for executing a new task. Denote the ready time of the task input data and the completion time for the $j$th task being offloaded (i.e., $\text{T}_{\sigma_{j}}$) as $t_{\text{ready}}^{j}\left(\bm{\sigma},\mathbf{p}\right)$ and $t_{\text{comp}}^{j}\left(\bm{\sigma},\mathbf{p}\right)$, respectively. Thus, $t_{\text{ready}}^{j}\left(\bm{\sigma},\mathbf{p}\right)$ is given by
\begin{equation}
t_{\text{ready}}^{j}\left(\bm{\sigma},\mathbf{p}\right)=\sum_{k\leq j}\frac{d_{\sigma_{k}}}{R\left(p_{\sigma_{k}}\right)},j=1,\cdots,N,
\label{readytime}
\end{equation}
which is the sum transmission time of tasks $\text{T}_{\sigma_{1}},\cdots,\text{T}_{\sigma_{j}}$.

The completion time of task $\text{T}_{\sigma_{j}}$ not only depends on the ready time of its task input data, but also couples with the completion time of the task being offloaded previously, i.e., $\text{T}_{\sigma_{j-1}}$. Consequently, $t_{\text{comp}}^{j}\left(\bm{\sigma},\mathbf{p}\right)$ can be expressed as the following recursive form:
\begin{equation}
t_{\text{comp}}^{j}\left(\bm{\sigma},\mathbf{p}\right)=\begin{cases}
t_{\text{ready}}^{j}\left(\bm{\sigma},\mathbf{p}\right) + d_{\sigma_{j}}c_{\sigma_{j}}f_{\rm{ser}}^{-1}, &j=1\\
\max\{t_{\text{ready}}^{j}\left(\bm{\sigma},\mathbf{p}\right),t_{\text{comp}}^{j-1}\left(\bm{\sigma},\mathbf{p}\right)\}\\ \ \ \ \ \ \ \ \ \ \ \ \ \ \ \ \ \ \ \ \ \ + d_{\sigma_j}c_{\sigma_j}f_{\rm{ser}}^{-1}, &j>1,
\end{cases}
\label{completionTime}
\end{equation}
where $d_{\sigma_{j}}c_{\sigma_{j}}f_{\rm{ser}}^{-1}$ is the mobile-edge execution time of task $\text{T}_{\sigma_{j}}$, and $t_{\text{comp}}^{N}\left(\bm{\sigma},\mathbf{p}\right)$ is the execution delay for $\mathcal{T}$.

The device energy consumption for offloading the $N$ computation tasks is simply the transmit energy consumption, which is independent with the task offloading scheduling decision and can be written as
\begin{equation}
E_{\text{tr}}\left(\mathbf{p}\right) = \sum_{i=1}^{N} p_{i}\cdot \frac{d_{i}}{R\left(p_{i}\right)}= \sum_{j=1}^{N} p_{\sigma_{j}}\cdot \frac{d_{\sigma_{j}}}{R\left(p_{\sigma_{j}}\right)}.
\end{equation}
Consequently, the joint task offloading scheduling and transmit power allocation problem can be formulated as
\begin{align}
&\mathcal{P}_{1}: \min_{\bm{\sigma},\mathbf{p}}\ t^{N}_{\text{comp}}\left(\bm{\sigma},\mathbf{p}\right) + \eta \cdot E_{\text{tr}}\left(\mathbf{p}\right)\label{objP1}\\
&\ \ \ \ \ \ \mathrm{s.t.\ }\ \sigma_{j}\in\{1,\cdots,N\}, \sigma_{i}\neq \sigma_{j}, j\neq i, \forall  i,j \label{constraintSigma} \\
&\ \ \ \ \ \ \ \ \ \ \ \ 0\leq p_{i}\leq p_{\max},i=1,\cdots,N\label{peakpowerConst}.
\end{align}
The objective function in $\mathcal{P}_{1}$ is the weighted sum of the execution delay and device energy consumption with $\eta$ (in $\text{sec}\cdot \text{J}^{-1}$) as the weighting factor, which is able to adjust the tradeoff between the execution delay and device energy consumption performance. (\ref{constraintSigma}) guarantees the task offloading scheduling decision is valid, and (\ref{peakpowerConst}) stands for the peak and non-negative transmit power constraints for the mobile device.

It can be easily noticed that $\mathcal{P}_{1}$ is an MINLP problem as both the integer variables $\bm{\sigma}$ and the continuous variables $\mathbf{p}$ need to be optimized, which is very challenging to solve. In principle, the optimal solution for $\mathcal{P}_{1}$ can be obtained via exhaustive search, which, however, is practically infeasible due to high complexity. For instance, there will be $20!\approx 2.43\times 10^{18}$ different permutations for $N=20$, and for each permutation, the optimal transmit power allocation needs to be determined. As a result, instead of finding the optimal solution, we will propose a low-complexity sub-optimal algorithm for $\mathcal{P}_{1}$ based on alternating minimization in the next section.

\section{Sub-optimal Joint Task Offloading Scheduling and Transmit Power Allocation}
In this section, we will propose a low-complexity sub-optimal algorithm for $\mathcal{P}_{1}$ based on the idea of optimizing the task offloading scheduling and the transmit power allocation alternately.

\subsection{Optimal Task Offloading Scheduling}
We first develop an optimal task offloading scheduling algorithm for a given transmit power allocation vector by invoking techniques from \emph{flow shop scheduling theory} \cite{Emmons2013}. To begin with, we introduce a notation for a type of two-machine flow shop scheduling problems in the following definition.
\begin{Def}
$F2|(perm),(pmtn)|C_{\max}$ denotes the type of flow shop scheduling problems that satisfy the following properties:
\begin{itemize}
\item[a)] There are $N$ independent jobs and each of them should be processed sequentially by two machines $M_{1}$ and $M_{2}$.

\item[b)] All jobs are available for $M_{1}$ from the beginning, and they will be available for $M_{2}$ immediately after the completion of the processing at $M_{1}$.

\item[c)] Each machine can only process one job at each time.

\item[d)] Permutation scheduling is considered, i.e., the jobs go through the two machines in the same order, which is encoded as $(perm)$ in the notation.

\item[e)] Non-preemptive processing is enforced, i.e., interruption is not allowed once the processing of a job starts, which is encoded as $(pmtn)$ in the notation.

\item[f)] Each machine has a buffer with infinite capacity to store the arrived but not yet processed jobs.

\item[g)] The goal is to minimize the makespan of the $N$ jobs, which is encoded as $C_{\max}$ in the notation.
\end{itemize}
\label{Def2stageFS}
\end{Def}

We find that the task offloading scheduling problem is a two-machine flow shop scheduling problem, as demonstrated in the follow lemma.
\begin{lma}
For a given transmit power allocation vector $\mathbf{p}$, $\mathcal{P}_{1}$ is an $F2|(perm),(pmtn)|C_{\max}$ type of flow shop scheduling problem.
\end{lma}
\begin{proof}
In the task offloading scheduling problem, each computation task needs to go through two procedures before its completion, which are the transmission of the task input data and the mobile-edge execution at the MEC server. Hence, the transmitter of the mobile device and the MEC server can be regarded as $M_{1}$ and $M_{2}$ in Definition \ref{Def2stageFS}, respectively. Thus, the job processing time in the two machines correspond to the transmission time and mobile-edge execution time, respectively. By analogy, it can be checked that the properties in Definition \ref{Def2stageFS} also holds for the task offloading scheduling problem. As a result, $\mathcal{P}_{1}$ with a given $\mathbf{p}$ is an $F2|(perm),(pmtn)|C_{\max}$ type of flow shop scheduling problem.
\end{proof}

Since the task offloading scheduling problem is an $F2|(perm),(pmtn)|C_{\max}$ type of flow shop scheduling problem, the \emph{Johnson's Algorithm} can be applied to find the optimal solution \cite{Emmons2013}. In the Johnson's Algorithm, the task set $\mathcal{T}$ is partitioned into two disjoint subsets $\mathcal{F}$ and $\mathcal{G}$, where $\mathcal{F}\triangleq\{\text{T}_{i}\in\mathcal{T} | \frac{d_{i}}{R\left(p_{i}\right)}< \frac{d_{i}c_{i}}{f_{\rm{ser}}}\}$ and
$\mathcal{G}\triangleq\big\{\text{T}_{i}\in\mathcal{T} | \frac{d_{i}}{R\left(p_{i}\right)}\geq \frac{d_{i}c_{i}}{f_{\rm{ser}}}\}$. The tasks in set $\mathcal{F}$ will be scheduled before those in $\mathcal{G}$ according to the ascending order of the value of $d_{i}\slash R\left(p_{i}\right),\text{T}_i\in\mathcal{F}$, and the tasks in set $\mathcal{G}$ will be scheduled according to the descending order of $d_{i}c_{i}\slash f_{\rm{ser}},\text{T}_i\in\mathcal{G}$. Details of the Johnson's Algorithm are summarized in Algorithm \ref{JRalgorithm}, where the computation overhead mainly comes from the sorting procedures in Line 4. Hence, if the \emph{Quicksort Algorithm} is used as the sorting algorithm, the complexity of Algorithm \ref{JRalgorithm} is $\mathcal{O}\left(N\log N\right)$ \cite{Cormen09}, i.e., the task offloading scheduling problem can be solved optimally in polynomial time.

\begin{algorithm}[t]
\caption{The Johnson's Algorithm For The Optimal Task Offloading Scheduling Decision}
\label{alg1}
\begin{algorithmic}[1]
\STATE \textbf{Input:} $\mathcal{T}$, $\mathbf{d}$, $\mathbf{c}$, $\mathbf{p}$, and $f_{\rm{ser}}$.
\STATE \textbf{Output:} $\bm{\sigma}^{\text{opt}}$.
\STATE Obtain set $\mathcal{F}$ and $\mathcal{G}$.
\STATE Sort the computation tasks in set $\mathcal{F}$ and set $\mathcal{G}$ according to the ascending order of the transmission time and the descending order of the execution time at the MEC server, respectively, i.e., $d_{\left[1\right]}\slash R\left(p_{\left[1\right]}\right)\leq \cdots \leq d_{\left[|\mathcal{F}|\right]}\slash R\left(p_{\left[|\mathcal{F}|\right]}\right),\text{T}_{\left[i\right]}\in\mathcal{F},i=1,\cdots,|\mathcal{F}|$ and $d_{\langle 1\rangle}c_{\langle 1 \rangle}f_{\rm{ser}}^{-1}\geq \cdots \geq d_{\langle |\mathcal{G}|\rangle}c_{\langle |\mathcal{G}| \rangle}f_{\rm{ser}}^{-1},\text{T}_{\langle j \rangle} \in \mathcal{G},j=1,\cdots,|\mathcal{G}|$.
\STATE $\bm{\sigma}^{\text{opt}}=\left[\left[1\right],\cdots,\left[|\mathcal{F}|\right],\langle 1\rangle,\cdots,\langle |\mathcal{G}|\rangle\right]$.
\end{algorithmic}
\label{JRalgorithm}
\end{algorithm}

\subsection{Optimal Transmit Power Allocation}

In this subsection, we investigate the optimal transmit power allocation for the computation tasks with a given task offloading scheduling decision $\bm{\sigma}$. Note that both the execution delay and the device energy consumption depend on the transmit power allocation vector $\mathbf{p}$, and $t_{\text{comp}}^{N}\left(\bm{\sigma},\mathbf{p}\right)$ is given in a recursive form, while the objective function in $\mathcal{P}_{1}$ given $\bm{\sigma}$ is non-differentiable, which makes it difficult to solve. To overcome this issue, we introduce a set of auxiliary variables $\tilde{\mathbf{t}}_{\bm{\sigma}}\triangleq \left[\tilde{t}_{\sigma_1},\cdots,\tilde{t}_{\sigma_N}\right]$ and formulate a modified version of the transmit power allocation problem as
\begin{align}
&\mathcal{P}_{2}: \min_{\tilde{\mathbf{t}}_{\bm{\sigma}},\mathbf{p}}\ \tilde{t}_{\sigma_{N}} + \eta \cdot \sum_{i=1}^{N} p_{\sigma_{i}}\frac{d_{\sigma_{i}}}{R\left(p_{\sigma_{i}}\right)}\\
&\ \ \ \ \ \ \mathrm{s.t.\ }\ (\ref{peakpowerConst}) \nonumber \\
&\ \ \ \ \ \ \ \ \ \ \ \ \tilde{t}_{\sigma_{i}}\geq \sum_{j=1}^{i}\frac{d_{\sigma_{j}}}{R\left(p_{\sigma_{j}}\right)}+\frac{d_{\sigma_{i}}c_{\sigma_{i}}}{f_{\rm{ser}}},i=1,\cdots,N\label{RelaxConstTX1}\\
&\ \ \ \ \ \ \ \ \ \ \ \ \tilde{t}_{\sigma_{i}}\geq \tilde{t}_{\sigma_{i-1}} + \frac{d_{\sigma_{i}}c_{\sigma_{i}}}{f_{\rm{ser}}},i=2,\cdots,N,\label{RelaxConstEX}
\end{align}
which is a relaxed version of $\mathcal{P}_{1}$ given $\bm{\sigma}$ as $\tilde{t}_{\sigma_{i}}\geq \max\{\sum_{j=1}^{i}\frac{d_{\sigma_{j}}}{R\left(p_{\sigma_{j}}\right)},\tilde{t}_{\sigma_{i-1}}\} + \frac{d_{\sigma_{i}}c_{\sigma_{i}}}{f_{\rm{ser}}},i \geq 2$ and $\tilde{t}_{\sigma_{1}}\geq \frac{d_{\sigma_{1}}}{R\left(p_{\sigma_{1}}\right)}+\frac{d_{\sigma_{1}}c_{\sigma_{1}}}{f_{\rm{ser}}}$. In the following lemma, we show such a relaxation is tight.

\begin{lma}
If $\langle \mathbf{p}^{\text{opt}},\tilde{\mathbf{t}}^{\text{opt}}_{\bm{\sigma}}\rangle$ is an optimal solution for $\mathcal{P}_{2}$, then $\mathbf{p}^{\text{opt}}$ is also optimal for $\mathcal{P}_{1}$ with a given task offloading scheduling decision $\bm{\sigma}$, and $\tilde{t}^{\text{opt}}_{\sigma_{N}}$ is the corresponding execution delay.
\end{lma}
\begin{proof}
With the optimal transmit power allocation vector $\mathbf{p}^{\text{opt}}$, we can construct a new feasible $\tilde{\mathbf{t}}'_{\bm{\sigma}}$ for $\mathcal{P}_{2}$ as
\begin{equation}
\tilde{t}_{\sigma_{i}}'=
\begin{cases}
\frac{d_{\sigma_{i}}}{R\left(p^{\text{opt}}_{\sigma_{i}}\right)}+\frac{d_{\sigma_{i}}c_{\sigma_{i}}}{f_{\rm{ser}}}, &i=1\\
\max\{\tilde{t}_{\sigma_{i-1}}',\sum_{j=1}^{i}\frac{d_{\sigma_{j}}}{R\left(p^{\text{opt}}_{\sigma_{j}}\right)}\} + \frac{d_{\sigma_{i}}c_{\sigma_{i}}}{f_{\rm{ser}}}, &i>1,
\end{cases}
\end{equation}
where $\tilde{t}'_{\sigma_{N}}$ is no larger than any of the feasible $\tilde{t}_{{\sigma}_{N}}$'s for $\mathcal{P}_{2}$ given $\mathbf{p}^{\text{opt}}$, i.e., $\tilde{t}^{\text{opt}}_{\sigma_{N}}=\tilde{t}'_{\sigma_{N}}$ and $\langle \mathbf{p}^{\text{opt}},\tilde{\mathbf{t}}_{\bm{\sigma}}'\rangle$ is also optimal for $\mathcal{P}_{2}$. Therefore, by applying $\mathbf{p}^{\text{opt}}$ to $\mathcal{P}_{1}$ given $\bm{\sigma}$, the value of the objective function $\text{Val}_{\mathcal{P}_{1}}$ equals $\text{Val}^{\text{opt}}_{\mathcal{P}_{2}}$, where $\text{Val}^{\text{opt}}_{\mathcal{P}_{2}}$ is the optimal value for $\mathcal{P}_{2}$. Since $\mathcal{P}_{2}$ is a relaxation of $\mathcal{P}_{1}$ given the task offloading scheduling decision, i.e., $\text{Val}_{\mathcal{P}_{1}}^{\text{opt}}\geq \text{Val}^{\text{opt}}_{\mathcal{P}_{2}}$ with $\text{Val}_{\mathcal{P}_{1}}^{\text{opt}}$ as the optimal value for $\mathcal{P}_{1}$ given $\bm{\sigma}$, we have $\text{Val}_{\mathcal{P}_{1}}=\text{Val}_{\mathcal{P}_{1}}^{\text{opt}}$, i.e., $\mathbf{p}^{\text{opt}}$ is optimal for $\mathcal{P}_{1}$ given $\bm{\sigma}$.
\end{proof}

Thus, we may concentrate on $\mathcal{P}_{2}$ in order to obtain the optimal transmit power allocation vector. Nevertheless, $\mathcal{P}_{2}$ is still difficult to solve due to its non-convexity, which is shown in the following lemma.
\begin{lma}
$\mathcal{P}_{2}$ is a non-convex optimization problem.
\label{nonconvexity}
\end{lma}
\begin{proof}
The proof can be obtained by verifying the concavity of $\frac{p}{R\left(p\right)}$ ($p> 0$), which is omitted for brevity.
\end{proof}

Fortunately, we are able to transform $\mathcal{P}_{2}$ into a convex optimization problem with the \emph{change-of-variable} technique similar to the one used in \cite{Lorenzo13}. In particular, by introducing a set of new variables $\xi_{\sigma_i}= 1\slash R\left(p_{\sigma_i}\right),i=1,\cdots,N$, $\mathcal{P}_{2}$ can be rephrased as
\begin{align}
&\mathcal{P}_{3}: \min_{\tilde{\mathbf{t}}_{\bm{\sigma}},\bm{\xi}}\ \tilde{t}_{\sigma_{N}} + C \cdot \sum_{i=1}^{N} d_{\sigma_{i}} \xi_{\sigma_{i}} \left(2^{\frac{1}{\omega \xi_{\sigma_{i}}}}-1\right)\\
&\ \ \ \ \ \ \mathrm{s.t.\ }\ (\ref{RelaxConstEX})\nonumber\\
&\ \ \ \ \ \ \ \ \ \ \ \ \xi_{\sigma_{i}} \geq D ,i=1,\cdots,N \label{boundXi}\\
&\ \ \ \ \ \ \ \ \ \ \ \ \tilde{t}_{\sigma_{i}}\geq \sum_{j=1}^{i}d_{\sigma_{j}}\xi_{\sigma_{j}}+\frac{d_{\sigma_{i}}c_{\sigma_{i}}}{f_{\rm{ser}}},i=1,\cdots,N,\label{RelaxConstTX2}
\end{align}
where $C\triangleq\frac{\eta  N_{0}\omega}{g_{0}\left(L_{0}\slash L\right)^{\theta}}$ and $D\triangleq \frac{1}{\omega\log_{2}\left(1+\frac{g_{0}\left(L_{0}\slash L\right)^{\theta}}{N_{0}\omega}p_{\max}\right)}$. In the following lemma, we show that $\mathcal{P}_{3}$ is a convex problem.
\begin{lma}
$\mathcal{P}_{3}$ is a convex optimization problem.
\end{lma}
\begin{proof}
Denote $\psi\left(\xi\right)\triangleq\xi\left(2^{\frac{1}{\omega \xi}}-1\right)$. The second-order derivative of $\psi\left(\xi\right)$ is given by
$\frac{d^{2}\psi\left(\xi\right)}{d \xi^{2}}=2^{\frac{1}{\omega\xi}}\frac{\ln^{2}2}{\omega^{2}\xi^{3}}\geq 0,\forall \xi >0$, i.e., $\psi\left(\xi\right)$ is a convex function of $\xi$ ($\xi>0$). Since the objective function is a summation of convex and linear functions, and the constraints are linear, $\mathcal{P}_{3}$ is a convex optimization problem \cite{Boyd04}.
\end{proof}

Hence, the optimal transmit power allocation vector can be obtained numerically using convex optimization solvers such as CVX \cite{CVX13}. Interestingly, we find that the optimal transmit power allocation vector, i.e., $p^{\text{opt}}_{\sigma_{1}},\cdots,p^{\text{opt}}_{\sigma_{N}}$, is a non-increasing sequence, as demonstrated in the following corollary.
\begin{corol}
Suppose $1\leq i< j \leq N$, then $p_{\sigma_{i}}^{\text{opt}}\geq p_{\sigma_{j}}^{\text{opt}}$, where $\mathbf{p}^{\text{opt}}$ is the optimal power allocation vector for $\mathcal{P}_{2}$ (also for $\mathcal{P}_{1}$ given $\bm{\sigma}$).
\label{nonincreasingProperty}
\end{corol}
\begin{proof} Let $\langle \bm{\xi}^{\text{opt}}, \tilde{\mathbf{t}}^{\text{opt}}_{\bm{\sigma}}\rangle$ be the optimal solution for $\mathcal{P}_{3}$. Since $\xi_{\sigma_{i}}=1\slash R\left(p_{{\sigma_{i}}}\right)$, it is equivalent to show $\xi_{\sigma_{i}}^{\text{opt}}\leq \xi_{\sigma_{j}}^{\text{opt}}, 1\leq i<j\leq N$. We first write the partial Lagrangian for $\mathcal{P}_{3}$ as
\begin{equation}
\begin{split}
&\mathcal{L}\left(\bm{\xi},\tilde{\mathbf{t}}_{\bm{\sigma}},\bm{\lambda}_{\bm{\sigma}}\right)= \tilde{t}_{\sigma_{N}} + C \cdot \sum_{i=1}^{N} d_{\sigma_{i}} \xi_{\sigma_{i}} \left(2^{\frac{1}{\omega \xi_{\sigma_{i}}}}-1\right)\\
&+ \sum_{i=1}^{N}\alpha_{\sigma_i}\left(\sum_{j=1}^{i}d_{\sigma_{j}}\xi_{\sigma_{j}}+\frac{d_{\sigma_{i}}c_{\sigma_{i}}}{f_{\rm{ser}}}-\tilde{t}_{\sigma_{i}}\right)\\
&+\sum_{k=2}^{N}\beta_{\sigma_k}\left(\tilde{t}_{\sigma_{k-1}}-\tilde{t}_{\sigma_{k}}+\frac{d_{\sigma_{k}}c_{\sigma_{k}}}{f_{\rm{ser}}}\right),
\end{split}
\end{equation}
where $\bm{\alpha}_{\bm{\sigma}}=\left[\alpha_{\sigma_{1}},\cdots,\alpha_{\sigma_{N}}\right]\succeq \mathbf{0}$ and $\bm{\beta}_{\bm{\sigma}}=\left[\beta_{\sigma_{2}},\cdots,\beta_{\sigma_{N}}\right]\succeq \mathbf{0}$ are the Lagrangian multipliers for constraints (\ref{RelaxConstTX2}) and (\ref{RelaxConstEX}), respectively. Define  $\bm{\lambda}_{\bm{\sigma}}\triangleq \left[\bm{\alpha}_{\bm{\sigma}},\bm{\beta}_{\bm{\sigma}}\right]$. Since $\mathcal{P}_{3}$ is a convex problem, we consider its dual problem
$\max_{\bm{\lambda}_{\bm{\sigma}}\succeq \mathbf{0}} g\left(\bm{\lambda}_{\bm{\sigma}}\right)$,
where $g\left(\bm{\lambda}_{\bm{\sigma}}\right)=\inf_{\tilde{t}_{\sigma_{i}},\xi_{\sigma_i}\geq D,\forall i}\mathcal{L}\left(\bm{\xi},\tilde{\mathbf{t}}_{\bm{\sigma}},\bm{\lambda}_{\bm{\sigma}}\right)$. If the optimal dual variables are given by $\bm{\lambda}_{\bm{\sigma}}^{\text{opt}}$, the optimal primal variables $\bm{\xi}^{\text{opt}}$ can be obtained by solving
\begin{equation}
\min_{\xi_{\sigma_{i}}\geq D,\forall i} C\cdot \sum_{i=1}^{N} d_{\sigma_{i}} \xi_{\sigma_{i}} \left(2^{\frac{1}{\omega \xi_{\sigma_{i}}}}-1\right) + \sum_{i=1}^{N}\alpha^{\text{opt}}_{\sigma_i}\left(\sum_{j=1}^{i}d_{\sigma_{j}}\xi_{\sigma_{j}}\right),
\end{equation}
for which the optimal solution is given by $\xi^{\star}_{\sigma_i}\left(\bm{\lambda}_{\bm{\sigma}}^{\text{opt}}\right)= \max\{D,\varphi\left(\sum_{j=i}^{N}\alpha^{\text{opt}}_{\sigma_{j}}\right)\}$. Here, when $x=0$, $\varphi\left(x\right)\triangleq +\infty$, and when $x>0$, $\varphi\left(x\right)$ denotes the root of equation $-d \psi\left(\xi_{\sigma_{i}}\right)\slash d \xi_{\sigma_{i}} = 1-2^{\frac{1}{\omega \xi_{\sigma_{i}}}}\left(1-\frac{\ln 2}{\omega \xi_{\sigma_{i}}}\right) = x\slash C$. Since $d \psi\left(\xi_{\sigma_{i}}\right)\slash d \xi_{\sigma_{i}}$ is an increasing function of $\xi_{\sigma_{i}}$ ($\xi_{\sigma_{i}}>0$), $\lim_{\xi_{\sigma_{i}}\rightarrow 0}d \psi\left(\xi_{\sigma_{i}}\right)\slash d \xi_{\sigma_{i}}=-\infty$ and $\lim_{\xi_{\sigma_{i}}\rightarrow +\infty}d \psi\left(\xi_{\sigma_{i}}\right)\slash d \xi_{\sigma_{i}}= 0$, we can show that for $i<j$, $\xi_{\sigma_{i}}^{\text{opt}}\leq \xi_{\sigma_{j}}^{\text{opt}}$ as $\alpha^{\text{opt}}_{\sigma_{l}}\geq 0,\forall l=1,\cdots,N$, which completes the proof.
\end{proof}

An intuitive explanation for Corollary \ref{nonincreasingProperty} is that the completion time of $\text{T}_{\sigma_{i}}$ will affect those of all the subsequent tasks, i.e., $\text{T}_{\sigma_{i+1}},\cdots,\text{T}_{\sigma_{N}}$. Thus, it is desirable to allocate a higher transmit power in order to reduce its transmission time.

\subsection{The Alternating Minimization Algorithm}

In the proposed algorithm, the task offloading scheduling decision and the transmit power allocation vector will be updated in an alternating manner, for which, the key steps are summarized in Algorithm \ref{AtlMinalgorithm}. Since the Johnson's Algorithm determines the optimal task offloading scheduling with a given transmit power allocation vector, and the solution for $\mathcal{P}_{3}$ offers the optimal transmit power allocation with a given task offloading scheduling decision, the value of the objective function in $\mathcal{P}_{1}$ decreases after each update, i.e., the convergence of Algorithm \ref{AtlMinalgorithm} is guaranteed.

\begin{algorithm}[h]
\caption{Sub-Optimal Joint Task Offloading Scheduling and Transmit Power Allocation Algorithm}
\label{alg2}
\begin{algorithmic}[1]
\STATE \textbf{Input:} $\mathbf{p}^{\text{prop}}=p_{\max}\mathbf{1}$, $\bm{\sigma}^{{\text{prop}}}=\left[1,\cdots,N\right]$, $\text{iter}_{\max}=50$, $\text{Val}_{\text{new}}=t^{N}_{\text{comp}}\left(\bm{\sigma}^{{\text{prop}}},\mathbf{p}^{{\text{prop}}}\right)+ \eta \cdot E_{\text{tr}}\left(\mathbf{p}^{{\text{prop}}}\right)$, $\text{Val}_{\text{old}}=\text{Val}_{\text{new}} + 10$, $\epsilon=10^{-7}$, and $I=0$.
\STATE \textbf{Output:} $\mathbf{p}^{\text{prop}}$ and $\bm{\sigma}^{\text{prop}}$.
\STATE \textbf{While} $\text{Val}_{\text{old}}-\text{Val}_{\text{new}} \geq \epsilon$ and $I\leq \text{iter}_{\max}$ \textbf{do}
\STATE \hspace{6pt} Set $I=I+1$.
\STATE \hspace{6pt} Set $\text{Val}_{\text{old}}=\text{Val}_{\text{new}}$.
\STATE \hspace{6pt} Update $\bm{\sigma}^{\text{prop}}$ using Algorithm \ref{JRalgorithm} with the transmit power\\
        \ \ \ allocation vector $\mathbf{p}^{{\text{prop}}}$.
\STATE \hspace{6pt} Update $\mathbf{p}^{\text{prop}}$ with the task offloading scheduling \\
\ \ \  decision $\bm{\sigma}^{\text{prop}}$ by solving $\mathcal{P}_{3}$.
\STATE \hspace{6pt} Set $\text{Val}_{\text{new}}=t^{N}_{\text{comp}}\left(\bm{\sigma}^{\text{prop}},\mathbf{p}^{\text{prop}}\right)+ \eta \cdot E_{\text{tr}}\left(\mathbf{p}^{\text{prop}}\right)$.
\STATE \textbf{End while}
\end{algorithmic}
\label{AtlMinalgorithm}
\end{algorithm}

\section{Simulation Results}
In this section, we first investigate the impact of the optimal task offloading scheduling on the delay performance. The proposed joint task offloading scheduling and transmit power allocation algorithm will then be evaluated. In simulations, the task input data size and the task workload are assumed to be uniformly distributed, i.e., $d_{i}\sim \text{Unif}\left(\left[0,2 d_{\text{avg}}\right]\right)$ and $c_{i}\sim \text{Unif}\left(\left[0,2 c_{\text{avg}}\right]\right)$ where $d_{\text{avg}}= 1\ \text{kbits}$ and $c_{\text{avg}} = 797.5\ \text{cycles}\slash\text{bit}$ \cite{Miettinen1006}. The CPU speed at the MEC server is set to be $f_{\text{ser}}= 1\ {\text{GHz}}$ unless otherwise specified.

\subsection{Impact of The Optimal Task Offloading Scheduling}
We investigate the impact of the optimal task offloading scheduling on the execution delay by setting $\eta = 0\ {\text{sec}\cdot \text{J}^{-1}}$, which is obtained by Algorithm \ref{alg1}. A random scheduling algorithm is adopted as the performance benchmark, where the task offloading scheduling decision $\bm{\sigma}$ is a random permutation of the task indices. Three scenarios with $R = \frac{f_{\text{ser}}}{c_{\text{avg}}}\approx 1.25\ \text{Mbps}$, $R=\frac{2f_{\text{ser}}}{c_{\text{avg}}}\approx 2.51\ \text{Mbps}$, and $R=\frac{f_{\text{ser}}}{2c_{\text{avg}}}\approx 0.63\ \text{Mbps}$, are considered, where the average transmission time of the task input data equals the average processing time at the MEC server in the first case.
\vspace{-10pt}
\begin{figure}[h]
\centering
\includegraphics[width=0.425\textwidth]{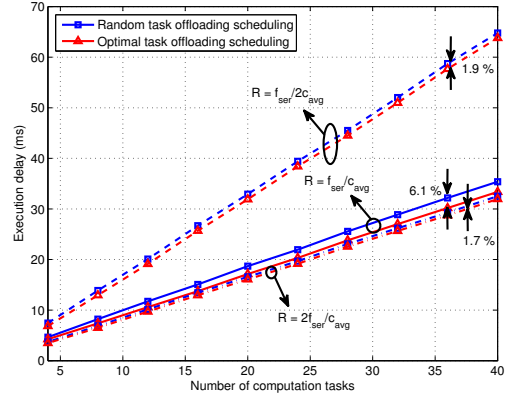}
\vspace{-10pt}
\caption{Execution delay vs. the number of computation tasks.}
\label{CTvsNumTasks}
\end{figure}

In Fig. \ref{CTvsNumTasks}, we show the relationship between the execution delay and the number of computation tasks. It can be seen that the execution delay increases linearly with the number of tasks, and a larger transmission rate results in a lower execution delay, which agree with intuitions. However, increasing the transmission rate from $0.63$ to $1.25\ \text{Mbps}$ brings a significant delay reduction, while the delay reduction by further doubling the transmission rate shrinks. This is because in the latter case, the system bottleneck turns into the limited computational resource at the MEC server. Besides, the proposed optimal task offloading scheduling algorithm outperforms the random task offloading scheduling for all cases, while the performance gain achieved by the optimal algorithm is most obvious when $R = \frac{f_{\text{ser}}}{c_{\text{avg}}}$ (e.g. 6.1\% for $N=35$). In other words, the optimal task offloading scheduling is more critical when the available radio resource and computational resource are relatively balanced, i.e., neither of them dominates the other.

\subsection{Joint Task Offloading Scheduling and Transmit Power Allocation}

We evaluate the proposed joint task offloading scheduling and transmit power allocation algorithm by setting $g_{0}=-40\ \text{dB}$, $L_{0}=1\ \text{m}$, $L=100\ \text{m}$, $\theta=4$, $\omega = 1\ \text{MHz}$, $N_{0}=-174\ {\text{dBm}}\slash {\text{Hz}}$, $p_{\max}=100\ {\text{mW}}$, and $N = 20$.

\begin{figure}[h]
\centering
\includegraphics[width=0.425\textwidth]{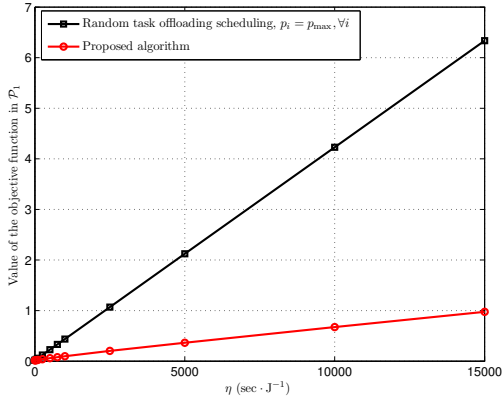}
\vspace{-10pt}
\caption{Value of the objective function in $\mathcal{P}_{1}$ vs. $\eta$.}
\label{ValObj}
\end{figure}

\begin{figure}[h]
\centering
\includegraphics[width=0.425\textwidth]{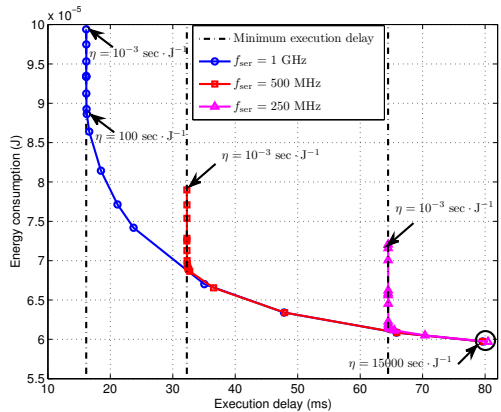}
\vspace{-10pt}
\caption{Device energy consumption vs. execution delay. (The device energy consumption with $p_{i}=p_{\max}$ is $4.21\times 10^{-4}\ \text{J}$.)}
\label{EnergyVSCompleTime}
\end{figure}

By varying $\eta$, the values of the objective function in $\mathcal{P}_{1}$ under the proposed algorithm and a benchmark algorithm with random task offloading scheduling and maximum transmit power are shown in Fig. \ref{ValObj}. From the curves, we see a significant performance improvement achieved by the proposed algorithm. Besides, the performance gain becomes more dramatic as $\eta$ increases since the transmit power allocation vector used in the benchmark scheme is non-adaptive to $\eta$.

The relationship between the device energy consumption and the execution delay is shown in Fig. \ref{EnergyVSCompleTime}, from which, we see the device energy consumption decreases while the execution delay increases with $\eta$. It is also shown that allocating the transmission energy beyond a threshold has no effect on improving the delay performance. This indicates that a large device energy consumption reduction can be achieved without loss of the optimal delay performance under the proposed algorithm with a suitable choice of $\eta$. For instance, with $\eta = 100\  \text{sec}\cdot \text{J}^{-1}$, $78\%$ of the device energy consumption is saved for $f_{\text{ser}}=1\ \text{GHz}$. This is due to the fact that the delay performance is limited by the computational resource, similar to what was observed in Fig. \ref{CTvsNumTasks}. In addition, when the device energy consumption becomes sufficiently small ($\eta$ becomes large enough), the curves for different CPU speeds at the MEC sever converge, i.e., the delay performance is constrained by the radio resource. This reveals a fundamental design principle for MEC systems: Once the system is constrained by the available radio resource, there is no need to deploy too much computational resources.

\section{Conclusions}
In this paper, we investigated joint task offloading scheduling and transmit power allocation for MEC systems with multiple independent tasks. Based on flow shop scheduling and convex optimization, we proposed a low-complexity sub-optimal algorithm to minimize the weighted sum of the execution delay and device energy consumption. It was found that the optimal task offloading scheduling achieves the most noticeable delay performance improvement when the available radio and computational resources are relatively balanced. Also, near-optimal delay performance together with a large device energy consumption reduction can be achieved by the proposed algorithm. For future investigation, it would be interesting to extend this work for mobile devices with certain computation capability, where the task offloading decision, i.e., whether to offload a task or not, the task offloading scheduling decision, as well as the transmit power allocation need to be jointly optimized.

\end{document}